\documentclass[11pt]{article}
\usepackage{amssymb,amsmath,amsthm,sectsty,url,ifplatform}
\usepackage[letterpaper,hmargin=0.972in,vmargin=1.0in]{geometry}

\ifwindows
    \usepackage[dvipdfm,bookmarks=true,pdfstartview=FitH,colorlinks,linkcolor=black,filecolor=black,citecolor=black,urlcolor=black]{hyperref}
\else
    \usepackage[bookmarks=true,pdfstartview=FitH,colorlinks,linkcolor=black,filecolor=black,citecolor=black,urlcolor=black]{hyperref}
\fi

\theoremstyle{plain}
\newtheorem{theorem}{Theorem}[section]
\newtheorem{lemma}[theorem]{Lemma}

\theoremstyle{definition}

\sectionfont{\large}
\subsectionfont{\normalsize}
\numberwithin{equation}{section} 

\newcommand{\Ex}{\mathbb{E}}
\newcommand{\floor}[1]{\llcorner #1 \lrcorner}
\newcommand{\ceil}[1]{\ulcorner #1 \urcorner}
\newenvironment{MyEqn}[1]{\setlength\arraycolsep{2pt}\begin{eqnarray*}
#1}{\end{eqnarray*}}%
\renewenvironment{proof}{\begin{trivlist} \item[\hspace{\labelsep}{\bf
\noindent Proof.\/}] }{\qedsymb\end{trivlist}}
\newcommand{\qedsymb}{\hfill{\rule{2mm}{2mm}}}

\newcommand{\Init}{{\sf Initialize}}
\newcommand{\Insert}{{\sf Insert}}
\newcommand{\Membership}{{\sf Membership}}
\newcommand{\Yes}{{\sf Yes}}
\newcommand{\No}{{\sf No}}

\begin{document}

\title{How to Approximate A Set \\ Without Knowing Its Size In Advance}
\author{Rasmus Pagh\thanks{IT University of Copenhagen. Email: {\tt pagh@itu.dk}.} \and Gil Segev\thanks{Stanford University,
Stanford, CA 94305, USA. Email: {\tt segev@stanford.edu}.} \and Udi Wieder\thanks{Microsoft Research Silicon Valley, Mountain View, CA 94043, USA. Email:  {\tt uwieder@microsoft.com}.}}
\date{}
\maketitle

\begin{abstract}
The dynamic approximate membership problem asks to represent a set $S$ of size $n$, whose elements are provided in an on-line fashion, supporting membership queries without false negatives and with a false positive rate at most $\epsilon$. That is, the membership algorithm must be correct on each $x \in S$, and may err with probability at most $\epsilon$ on each $x \notin S$.

We study a well-motivated, yet insufficiently explored, variant of this problem where the size $n$ of the set is not known in advance. Existing optimal approximate membership data structures require that the size is known in advance, but in many practical scenarios this is not a realistic assumption.
Moreover, even if the eventual size $n$ of the set is known in advance,
it is desirable to have the smallest possible space usage also when the current number of inserted elements is smaller than $n$.
Our contribution consists of the following results:
\begin{itemize}
\item We show a super-linear gap between the space complexity when the size is known in advance and the space complexity when the size is not known in advance. When the size is known in advance, it is well-known that $\Theta(n \log(1/\epsilon))$ bits of space are necessary and sufficient (Bloom '70, Carter et al.\ '78). However, when the size is not known in advance, we prove that at least $(1 - o(1))n \log(1/\epsilon) + \Omega(n \log \log n)$ bits of space must be used. In particular, the average number of bits per element must depend on the size of the set.

\item We show that our space lower bound is tight, and can even be matched by a highly efficient data structure. We present a data structure that uses $(1+o(1)) n\log(1/\epsilon) \allowbreak +  O(n\log\log n)$ bits of space for approximating any set of any size $n$, without having to know $n$ in advance. Our data structure supports membership queries in constant time in the worst case with high probability, and supports insertions in expected amortized constant time. Moreover, it can be ``de-amortized'' to support also insertions in constant time in the worst case with high probability by only increasing its space usage to $O(n\log(1/\epsilon) +  n\log\log n)$ bits.
\end{itemize}%

\end{abstract}

\thispagestyle{empty}
\newpage \setcounter{page}{1}

\section{Introduction}

Dictionaries play a fundamental role in the design and analysis of algorithms, enabling representation of any given set $S$ while supporting membership queries. For sets of size $n$ that are taken from a universe $U$ of size $u$, any dictionary must clearly use at least $\log \binom{u}{n} = n\log(u/n)+\Theta(n)$ bits of space\footnote{Throughout this paper all logarithms are to the base $2$.}. Whereas dictionaries offer {\em exact} representations of sets, in many realistic scenarios it is desirable to trade exact representations with {\em approximate} ones in order to reduce space consumption. This was observed already by Bloom \cite{Bloom70}, whose classical design of a Bloom Filter provides a simple and practical alternative to dictionaries.

Bloom's data structure solves the problem known these days as the {\em approximate membership} problem. This problem asks to represent any given set $S$ of size $n$ while supporting membership queries without false negatives, and with a false positive rate at most $\epsilon$. That is, the membership algorithm must be correct on any $x \in S$, and may err with probability at most $\epsilon$ on any $x \in U \setminus S$ (where the probability is taken over the randomness used by the data structure). The approximate membership problem can be considered in the {\em static} setting where the set is specified in advance, or in the {\em dynamic} setting where the elements of the set are specified one by one in an on-line fashion.

Bloom's data structure uses only $\log e \cdot n \log(1/\epsilon)$ bits of space (and solves the problem even in the dynamic setting), and Carter et al.\ \cite{CarterFGMW78} proved that this is essentially optimal: Any approximate membership data structure must use at least $n \log(1/\epsilon)$ bits of space, even in the static setting. Over the years a long line of research has shown how to design approximate membership data structures that are essentially optimal in both their space utilization and efficiency of their operations. We refer the reader to the survey of Broder and Mitzenmacher \cite{BroderM03} for various applications for approximation membership data structures, and to Section \ref{SubSec:RelatedWork} for an overview of the known results.

\paragraph{Approximating sets of unknown sizes.}  The vast majority of existing approximate membership data structures require that the size $n$ of the set $S$ to be approximated will be known in advance. In many practical scenarios, however, it is unrealistic to assume that the size is known in advance~\cite{GuoWCL06}.
Moreover, even if the eventual size $n$ of the set is known in advance,
it is desirable to have the smallest possible space usage also when the current number of inserted elements is smaller than $n$.

In this paper we study the well motivated, yet insufficiently explored, variant of the dynamic approximate membership problem where the size $n$ of the set is not known in advance. We refer to this problem as {\em approximate membership for sets of unknown sizes}. This problem is parameterized by $u \in \mathbb{N}$ and $0 < \epsilon < 1$, and asks to design a data structure offering three algorithms: $\Init$, $\Insert$, and $\Membership$. Upon initialization via the $\Init$ algorithm, the data structure is presented with a sequence of elements that are taken from a universe $U$ of size $u$. The elements are specified in an on-line fashion, and each element is processed using the $\Insert$ algorithm that updates the internal state of the data structure. The $\Membership$ algorithm should satisfy the following two requirements:
\begin{itemize}
\item {\bf No false negatives:} For any $n \leq u$, $S \subseteq U$ of size $n$, and $x \in S$, the $\Membership$ algorithm always outputs $\Yes$ on $x$ after the elements of $S$ are processed by the $\Insert$ algorithm.

\item {\bf False positive rate at most $\epsilon$:} For any $n \leq u$, $S \subseteq U$ of size $n$, and $x \notin S$, the $\Membership$ algorithm outputs $\Yes$ on $x$ with probability at most $\epsilon$ after the elements of $S$ are processed by the $\Insert$ algorithm (where the probability is taken over the randomness of the data structure).
\end{itemize}

\paragraph{Gradually-increasing space consumption.} For the approximate membership problem when the size $n$ of the set is known in advance, it is well-known that $\Theta(n \log(1/\epsilon))$ bits of space suffice even in the dynamic setting, and are essential even in the static setting (recall that a Bloom filter uses $O(n \log(1/\epsilon))$ bits which is asymptotically optimal). That is, the average number of bits for representing each element is $\Theta(\log(1/\epsilon))$ which is independent of the size of the set.

In this light, a natural question is whether this is also the case when the size $n$ is not known in advance, and the data structure is required to work for sets of any size $n \leq u$. That is, we ask the following question: Is there a dynamic approximate membership data structure that uses space $O(n \log(1/\epsilon))$ for representing any set $S$ of any size $n \leq u$? Somewhat surprisingly, this question was so far addressed only from a practical perspective, and has not been investigated from a foundational perspective. Moreover, the data structures we could find in the literature~\cite{AlmeidaBPH07, HaoKL08, GuoWCL06, GuoWCYL10, WeiJZF11, WeiJZF13} use space $\Omega(n\log n)$ bits (and query time $\Omega(\log n)$ or $\Omega(\log(1/\epsilon))$). These solution are somewhat naive from an algorithmic point of view, and provide poor asymptotic bounds.

\subsection{Our Contributions}

We present a lower bound and matching upper bounds on the space complexity of approximate membership for sets of unknown sizes. Our lower bound shows that if the size $n$ of the size of the sets to be approximated is not known in advance, then it is not possible to use an average of $O(\log(1/\epsilon))$ bits per elements as in the standard case. Specifically, we show a super-linear gap between the space complexity when $n$ is known in advance and the space complexity when $n$ is not known in advance. We prove the following theorem:

\begin{theorem}[Lower bound -- informal]\label{Thm:LowerBoundInformal}
Any data structure for approximate membership for sets of unknown sizes with false positive rate $\epsilon$ must use space $(1 - o(1))n \log(1/\epsilon) + \Omega(n \log \log n)$ bits after some number of insertions $n>u^{\delta}$, for any arbitrary small constant $0 < \delta < 1$.
\end{theorem}

In particular, Theorem \ref{Thm:LowerBoundInformal} states that the average number of bits per element must be at least $(1 - o(1))\log(1/\epsilon) + \Omega(\log\log n)$ at some point in time while processing a not-too-short sequence. We emphasize that in many practical scenarios (see \cite{BroderM03}) a typical false positive rate is a not-too-small constant (e.g., $\epsilon = 1/10$). For such a range of parameters our lower bound states that the average number of bits per element must be $\Omega(\log \log n)$ as opposed to constant.

We then show that our lower bound is asymptotically tight by presenting two constructions with a space usage that matches our lower bound up to additive lower order terms. We prove the following theorem:

\begin{theorem}[Upper bound -- informal]\label{Thm:UpperBoundInformal}
There exists a data structure for approximate membership for sets of unknown sizes with false positive rate $\epsilon$ that uses space $(1+o(1)) n\log(1/\epsilon) \allowbreak +  O(n\log\log n)$ bits for any sequence of $n>u^{\delta}$ insertions, for any arbitrary small constant $0 < \delta < 1$.
\end{theorem}

Our first construction (which can be viewed as a warm-up) is quite natural and uses a sequence of dynamic approximate membership data structures of geometrically-increasing sizes. It supports insertions in expected amortized constant time, but membership queries are supported in time $O(\log n)$. Our second construction is significantly more subtle, showing that in fact our space lower bound can be matched by a highly efficient data structure supporting membership queries in constant time in the worst case with high probability (while still enjoying expected amortized constant insertion time as in our first construction). Moreover, we show that it can be ``de-amortized'' to support also insertions in constant time in the worst case with high probability by increasing its space usage from $(1+o(1)) n\log(1/\epsilon) +  O(n\log\log n)$ bits to $O(n\log(1/\epsilon) +  n\log\log n)$ bits (with a rather small leading constant). We refer the reader to Section \ref{Subsec:Overview} for a high-level overview of the main ideas underlying our lower bound and constructions.

Finally, we note that in both our lower bound and constructions we consider approximate representation of sets whose size $n$ is polynomially related to the universe size $u$ (i.e., $n>u^{\delta}$ for any arbitrary small constant $0 < \delta < 1$). This is rather standard for exact or approximate representation of sets as one can always apply a universe reduction via simple universal hashing given {\em any} polynomial upper bound on the number of elements.

%

\subsection{Related work}\label{SubSec:RelatedWork}


\paragraph{Bloom filters.}
The elegant data structure proposed by Bloom~\cite{Bloom70} naturally allows dynamic insertions, but uses space that is a factor $\log e\approx 1.44$ larger than the information theoretic lower bound of $n\log(1/\epsilon)$ bits~\cite{CarterFGMW78}.
Another thing to notice is that Bloom filters do not allow deletions from $S$, as setting any bit to $0$ could result in false negatives.

Deletion queries can be supported by using {\em counting\/} Bloom filters~\cite{FanCAB00}, at the cost of an $\Omega(\log\log n)$ factor  increase in space usage.
Deletions are supported in the sense that the data structure will work correctly if no attempt is made to delete a false positive, but by definition it is not possible to prevent such deletions.
Cohen and Matias~\cite{CohenM03} present a way of decreasing the space overhead to $O(n)$ bits, and generalize approximate membership to approximate multiplicity in a multiset.

\paragraph{Dictionary-based approximate membership.}
Already in 1978, Carter et al.~\cite{CarterFGMW78} had presented a technique that would lead to a similar result.
They observed that maintaining the multiset $h(S)$, where $h: [u]\rightarrow [n/\epsilon]$ is a universal hash function~\cite{CarterW79}, yields a solution with space $n\log(1/\epsilon)+O(n)$ bits if the set $h(S)$ is stored in space close to the information theoretic bound of $\log\binom{n/\epsilon+n}{n}$ bits.
If deletions are not needed it suffices to store the set of distinct hash values $h(S)$.
This dynamic set can be stored succinctly with all operations taking $O(1)$ time with high probability~\cite{ArbitmanNS10}.
Dynamic multisets, and thus deletions, can be supported via a reduction to the standard membership problem~\cite{PaghPR05}, at the cost of amortized expected update bounds.
A more practical alternative was explored in~\cite{BonomiMPSV06}.

\paragraph{Separation of on-line and off-line space requirements.}
Dietzfelbinger and Pagh~\cite{DietzfelbingerP08} showed how to approach the $n\log(1/\epsilon)$
space lower bound up to a $o(n)$ term using query time $\omega(\log(1/\epsilon))$, in the case where $\epsilon$ is an integer power of $2$.
Independently, Porat~\cite{Porat09} achieved the same result with constant query time. Recently, Bellazougui and Venturini~\cite{BelazzouguiV13} showed how to eliminate the restriction on $\epsilon$, still maintaining constant query time.

Lovett and Porat~\cite{LovettP10} showed that these results for the static case do not extend to the situation where dynamic updates are allowed:
An overhead of $\Omega(n/\log(1/\epsilon))$ bits is required.
The lower bound holds even if there are no queries before the end of the insertion sequence.
In other words, this result implies that to build an approximate membership data structure for a key set given as a data stream, it does not suffice to use space close to the static size lower bound.

\paragraph{Dynamic space usage.}
The setting where space must depend on the current size of the set is more demanding from an upper bound perspective.
In fact, the techniques for this problem that we could find in the literature~\cite{AlmeidaBPH07, HaoKL08, GuoWCL06, GuoWCYL10, WeiJZF11, WeiJZF13} lead to $\Omega(\log n)$ or $\Omega(\log(1/\epsilon))$ query time, and a space overhead of $\Omega(n\log n)$ bits.
These data structures share the idea of working with a sequence of approximate membership data structures, all of which are queried.
If geometrically increasing capacity is chosen this means that there will be $\Omega(\log n)$ such data structures (of course, if we have some initial capacity $n_0$ this number decreases to $O(\log(n/n_0))$, which might be fine in practical situations -- but it is not asymptotically optimal).
A consequence of working with a series of approximations is that the sum of corresponding false positive rates $\epsilon_1,\epsilon_2,\dots$ must converge to~$\epsilon$. For example, In~\cite{AlmeidaBPH07} it is suggested to achieve this by letting~$\epsilon_i$ decrease geometrically with $i$.
This implies that $\epsilon_i = n^{-{\Omega(i)}}$, yielding $\Omega(n\log n)$ space usage.

\paragraph{Dynamic perfect hashing and retrieval.}
An approach to approximate membership in the static case is to store a perfect hash function that maps keys injectively to $\{1,\dots,n\}$, and then store a signature of $\log(1/\epsilon)$ bits for each key in an array, placed according to the perfect hash function.
More generally, a dynamic data structure for retrieval (e.g. the {\em Bloomier filters\/} of~\cite{ChazelleKRT04}) allows us to make a dynamic approximate membership data structure.
As shown in~\cite{MortensenPP05} both these problems require space $\Theta(n\log\log n)$ in the on-line setting.
However, the upper bounds have a fixed space usage (up to constant factors), and hence do not allow the kind of result we obtain.

\subsection{Overview of Our Contributions}\label{Subsec:Overview}

In this section we provide a high-level overview of the main ideas underlying our lower bound and constructions.

\paragraph{The lower bound: From Approximate membership to compression.} When dealing with dictionaries (i.e., with {\em exact} membership as opposed to {\em approximate} membership), it is quite simple to deal with the fact that the size of the set $S$ to be stored is not known in advance. Specifically, at any point in time a dictionary stores a description of the set $S$ of elements that were inserted so far. Then, upon inserting a new element $x \in U \setminus S$ this description can simply be updated to that of $S'  = S \cup \{x\}$. The dictionary can describe $S$ and then $S'$ using the minimal, information-theoretic, number of bits. Moreover, there are even time-efficient solutions that gradually increase the size of the dictionary while offering constant-time operations in the worst case together with an asymptotically-optimal space consumption at any point in time (see, for example, the work of Dietzfelbinger and auf der Heide \cite{DietzfelbingerMadH90}).

When dealing with {\em approximate} membership, however, it seems significantly more challenging when the size of the set $S$ to be approximated is not known in advance. For simplifying the following discussion we consider here {\em deterministic} approximate membership data structures, but note that the exact same ideas carry over to randomized ones. Specifically, for using asymptotically optimal space, an approximate membership data structure cannot afford to store an exact description of the set $S$ of elements that were inserted so far. Instead, any particular state of the data structure may be used for many sets other than $S$, and the result of the $\Membership$ algorithm must be $\Yes$ on any element that belongs to the union $\hat{S}$ of these sets.
Upon inserting a new element $x \in U \setminus S$, the data structure has to update the description of the current superset $\hat{S}$ to the description of some superset $\hat{S'}$ of $S' = S \cup \{x\}$. Note, however, that the data structure does not have access to the set $S$, but only to the approximation $\hat{S}$ containing $S$. Therefore, it must hold that $\hat{S} \subseteq \hat{S'}$ as 
any element of $\hat{S}$ might have been inserted, and false negatives are not allowed.

The main observation underlying our lower bound proof is that not only the new superset $\hat{S'}$ has to be larger than the old one $\hat{S}$ (as $\hat{S} \subseteq \hat{S'}$), but it actually has to be significantly larger. That is, upon the insertion of an element, the data structure must update its internal state by adding {\em many} elements to the currently stored superset. This is in contrast to the setting of {\em exact} membership discussed above, where upon the insertion of an element, a dictionary can update its internal state by adding only the newly added element to the currently stored set.  We formalize this observation via a compression argument showing that if $\hat{S'} \setminus \hat{S}$ is rather small, then we can ``compress'' the set $S'$ below the information-theoretic lower bound. We note that this argument takes into account only space utilization, and does not need to make any assumptions on the efficiency of the data structure in terms of the time complexity of its $\Insert$ and $\Membership$ algorithms. We refer the reader to Section \ref{Sec:LowerBound} for the proof of our lower bound.

\paragraph{Construction 1 (warm-up): Geometrically-increasing data structures.} Our first construction is quite natural and uses a sequence of dynamic approximate membership data structures of geometrically-increasing sizes. When viewing the sequence of inserted elements as consecutive subsequences, where the $i$th subsequence consists of $2^i$ elements, at the beginning of the $i$th subsequence we allocate a dynamic approximate membership data structure $\mathcal{B}_i$ with a false positive rate $\epsilon_i = \Theta (\epsilon / i^2)$. The elements of the $i$th subsequence are processed by $\mathcal{B}_i$. A membership query for an element $x \in U$ is performed by invoking the membership algorithm of each of the existing data structures $\mathcal{B}_i$, and reporting $\Yes$ if any of them does. Clearly, the construction has no false negatives, and its false positive rate is at most $\sum_{i = 1}^{\infty} \epsilon_i \leq \epsilon$.

By carefully instantiating the underlying $\mathcal{B}_i$'s with existing dynamic approximate membership data structures, for any sequence of $n$ insertions the data structure uses only $(1+o(1)) n \log(1/\epsilon) + O(n \log \log n)$ bits of space, and insertions are performed in expected amortized constant time. However, membership queries require time $\Theta(\log n)$ after $n$ insertions, as a separate membership query is needed for each of the existing $\mathcal{B}_i$'s.\footnote{We note that these $\Theta(\log n)$ membership queries can be executed in parallel.} We refer the reader to Section \ref{Sec:Construction1} for more details.

\paragraph{Construction 2: Constant-time operations.} Whereas our first construction is somewhat naive, our second construction is significantly more subtle, supporting membership queries in constant time in the worst case with high probability (while still enjoying expected amortized constant insertion time as in our first construction). Moreover, we show that it can be ``de-amortized'' to support also insertions in constant time in the worst case by increasing its space usage from $(1+o(1)) n\log(1/\epsilon) +  O(n\log\log n)$ bits to $O(n\log(1/\epsilon) +  n\log\log n)$ bits (with a rather small leading constant).

Unlike our first construction, this construction consists of only one data structure at any point in time. This data structure is a dynamic dictionary (i.e., an exact representation of a set) that is used for storing a carefully chosen superset of the elements that were inserted so far. For describing the main ideas underlying this construction, we again view the sequence of inserted elements as consecutive subsequences, where the $i$th subsequence consists of $2^i$ elements. The construction is initialized by sampling a function $h : U \rightarrow \{0,1\}^{\ell}$ from a pairwise independent collection of functions, where $\ell \ge \lceil \log(1/\epsilon) \rceil + \log u + 2$ (recall that $\epsilon$ is the required false positive rate, and that $u$ is that size of the universe of elements).

The basic idea is that for inserting an element $x$ as part of the $i$th subsequence, we store in the current dictionary $\mathcal{D}_i$ the value $h_i(x)$ that is defined as the leftmost $\ell_i = \lceil \log(1/\epsilon) \rceil + i + 2$ bits of $h(x)$. At the end of the $i$th subsequence, we transition from the current dictionary $\mathcal{D}_i$ to a newly allocated dictionary $\mathcal{D}_{i+1}$, and de-allocate the space used by $\mathcal{D}_i$. The transition is performed as follows: As $\mathcal{D}_i$ is a dictionary, we can enumerate all of its stored values, and for each such value $y \in \{0,1\}^{\ell_i}$ we insert both $y0 \in \{0,1\}^{\ell_{i+1}}$ and $y1 \in \{0,1\}^{\ell_{i+1}}$ to the new dictionary $\mathcal{D}_{i+1}$. Note that $\mathcal{D}_i$ stores $\ell_i$-bit values, and $\mathcal{D}_{i+1}$ stores $\ell_{i+1}$-bit values. The key point is that at any point in time there is only one dictionary $\mathcal{D}_i$, and therefore any membership query requires executing only one such query: Given an element $x$ and given that the current dictionary is $\mathcal{D}_i$, we execute a membership query for $h_i(x)$ in $\mathcal{D}_i$. Therefore, the time for supporting membership queries is identical to that of the underlying dictionaries.

This approach, however, needs to be refined as the number of stored values increases too fast. To match our lower bound, we would like to argue that each dictionary $\mathcal{D}_i$ stores $O(2^i)$ values. This is not the case: Each of the $2^j$ elements that are inserted as part of the $j$th subsequence, for any $j < i$, ``contributes'' $2^{i - j}$ values to $\mathcal{D}_i$, and therefore the number of values stored by $\mathcal{D}_{\log n}$ would be $O(n \log n)$ instead of $O(n)$.

We resolve this difficulty as follows. For inserting an element $x$ as part of the $i$th subsequence, we store in $\mathcal{D}_i$ the pair $(h_i(x), g_i(x))$, where $h_i(x)$ is defined as the leftmost $\ell_i = \lceil \log(1/\epsilon) \rceil + i + 2$ bits of $h(x)$ (as before), and $g_i(x)$ is defined as the next $r = \lceil \log \log u \rceil $ output bits of $h(x)$ (padded with the symbol $\bot$ when less than $r$ such bits are available)\footnote{Such a pair $(h_i(x), g_i(x))$ is inserted using $h_i(x)$ as its key, for enabling constant-time membership queries.}. The transitioning from $\mathcal{D}_i$ to $\mathcal{D}_{i+1}$ is now performed as follows: For each pair $(y, \alpha_1 \cdots \alpha_r) \in \{0,1\}^{\ell_{i}} \times \{0,1, \bot\}^r$ that is stored in $\mathcal{D}_i$ we insert to $\mathcal{D}_{i+1}$ either the pair $(y \alpha_1, \alpha_2 \cdots \alpha_r \bot)$ if $\alpha_1 \neq \bot$ (using $y \alpha_1$ as its key), or the two pairs $(y0, \alpha)$ and $(y1, \alpha)$ if $\alpha_1 = \bot$ (using $y0$ and $y1$ as the respective keys). This way, each of the $2^j$ elements that are inserted as part of the $j$th subsequence, for any $j < i$, ``contributes'' only $2^{i - j - r}$ values to $\mathcal{D}_i$, which guarantees that each $\mathcal{D}_i$ stores only $O(2^i)$ values. This, combined with a standard bucketing argument, enables us to match our space lower bound of using $(1+o(1)) n\log(1/\epsilon) \allowbreak +  O(n\log\log n)$ bits for any sequence of $n$ insertions. Note that this method has no false negatives, and we show that our choice of parameters guarantees that the false positive rate is at most $\epsilon$. Moreover, the time for supporting membership queries is identical to that of the underlying dictionaries.

The construction enjoys a good amortized insertion time: Most insertions correspond to standard insertions for the current $\mathcal{D}_i$, while only a small number of insertions require transitioning from $\mathcal{D}_i$ to $\mathcal{D}_{i+1}$. Specifically, we show that if the underlying dictionaries support insertions in expected amortize constant time, then so does our construction. Moreover, we also show that if the underlying dictionaries offer a constant insertion time in the worst case with high probability, then our construction can be modified to offer constant time insertions in the worst case with high probability. This follows the de-amortization technique of Arbitman, Naor and Segev \cite{ArbitmanNS09, ArbitmanNS10}, and only increases the space usage from $(1+o(1)) n\log(1/\epsilon) +  O(n\log\log n)$ bits to $O(n\log(1/\epsilon) +  n\log\log n)$ bits (with a rather small leading constant).

Finally, we also show that our construction can even support deletions, as long as no false positives are deleted.  We refer the reader to Section \ref{Sec:Construction2} for more details.

\subsection{Paper Organization}

The remainder of this paper is organized as follows. In Section \ref{Sec:Preliminaries} we present some essential preliminaries. In Section \ref{Sec:LowerBound} we prove our lower bound. In Sections \ref{Sec:Construction1} and \ref{Sec:Construction2} we present our constructions. Finally, in Section \ref{Sec:FutureResearch} we discuss various directions for future research.

\section{Preliminaries}\label{Sec:Preliminaries}

\paragraph{Notation.} For an integer $n \in \mathbb{N}$ we denote by $[n]$
the set $\{1, \ldots, n\}$. For a random variable $X$ we denote by $x \leftarrow X$
the process of sampling a value $x$ according to the distribution of $X$.
Similarly, for a finite set $S$ we denote by $x \leftarrow S$ the process
of sampling a value $x$ according to the uniform distribution over $S$.

\paragraph{Computational model.} We consider the unit cost RAM model in which the elements are taken from a
universe of size $u$, and each element can be stored in a single word of length $w = \lceil \log u
\rceil$ bits. Any operation in the standard instruction set can be executed in
constant time on $w$-bit operands. This includes addition, subtraction, bitwise Boolean operations,
left and right bit shifts by an arbitrary number of positions, and multiplication. The unit cost RAM model has been the subject of much research, and is considered the standard model for analyzing the efficiency of data structures (see, for example, \cite{DietzfelbingerP08,Hagerup98,HagerupMP01,Miltersen99,PaghP08,RamanR03} and the references therein).

\paragraph{$\boldsymbol{k}$-Wise independent functions.} A collection $\mathcal{H}$ of functions $h
: U \rightarrow V$ is $k$-wise independent if for any distinct $x_1, \ldots, x_k \in U$ and for any
$y_1, \ldots, y_k \in V$ it holds that
\[ \Pr_{h \leftarrow \mathcal{H}} \left[ h(x_1) = y_1 \wedge \cdots \wedge h(x_k) = y_k \right] = \frac{1}{|V|^k} . \]%
More generally, a collection $\mathcal{H}$ is $k$-wise $\delta$-dependent if for any distinct $x_1,
\ldots, x_k \in U$ the distribution $(h(x_1), \ldots, h(x_k))$ where $h$ is sampled from
$\mathcal{H}$ is $\delta$-close in statistical distance to the uniform distribution over $U^k$.

\section{The Lower Bound: From Approximate Membership to Compression}\label{Sec:LowerBound}

Let $\mathcal{D}$ be an approximate membership data structure for sets of unknown size, for a universe $U$ of size $u$ with a false positive rate $0 < \epsilon < 1$. Our lower bound holds for any such data structure that supports insertions and membership queries. We assume $\mathcal{D}$ has access to a read-only array of random bits at no cost in space, allowing randomized data structures. Since we do not limit the time complexity of the $\Membership$ algorithm, and all possible histories of the data structure can be computed using the random bits array, we may without loss of generality assume that $\mathcal{D}$ answers $\Yes$ on input $x$ exactly when the current state of the data structure is consistent with {\em some\/} history in which $x$ was inserted using the current array of random bits.

In this section we prove a lower bound on the space usage of $\mathcal{D}$ even if the size of the sets to be approximated is known to be in a certain interval (this only strengthens the lower bound). Specifically, we prove the following theorem:

\begin{theorem}\label{thm:lb}
Let $\mathcal{D}$, $U$, and $\epsilon$ be as above, and let $n \leq \epsilon u$ be sufficiently large and $1/\sqrt n \leq \alpha < 1$. If for any sequence of insertions of any length $m$ such that $\alpha n < m < n$, the data structure $\mathcal{D}$ uses at most $\beta m$ bits of space, then for any integer $\gamma \geq 2$ it holds that
\[ \beta \geq \left( 1-\frac{1}{\gamma} \right) \cdot \left( \log(1/\epsilon) + (1-9\epsilon) \log\log_\gamma (1/\alpha)  - \Theta(1) \right) . \]
\end{theorem}


In particular, by setting $\alpha = 1/\sqrt{n}$ and $\gamma = 2^{(\log n)^{\eta}}$, for some constant $0 < \eta < 1$, we obtain the lower bound of $(1 - o(1))n \log(1/\epsilon) + \Omega(n \log \log n)$ bits that is stated in Theorem \ref{Thm:LowerBoundInformal}. We note that asking that the data structure is space efficient only for sequences of at least $\alpha n$ elements can be viewed as allowing the data structure to process the first $\alpha n$ elements in an off-line manner using arbitrary space (which, again, only strengthens the lower bound). We also note that setting $\alpha = 1$ corresponds to the case where the size $n$ of the set is known in advance.

The proof of Theorem \ref{thm:lb} consists of two parts. In the first part (see Section \ref{Subsec:RandToDet}) we show that it suffices to prove the lower bound for deterministic data structures, where the probability of false positives is taken over the choice of a uniformly sampled element (instead of over the internal randomness of the data structure). This is a standard averaging argument showing that one can fix the randomness of any randomized data structure, without significantly increasing the false positive rate. In the second part (see Section \ref{Subsec:Rand}), we then follow the overview discussed in Section \ref{Subsec:Overview} for proving the lower bound for deterministic data structures.


\subsection{From Randomized to Deterministic Approximate Membership}\label{Subsec:RandToDet}

For any (sufficiently long) string $r \in \{0,1\}^*$, we denote by $\mathcal{B}_r$ the deterministic data structure obtained by fixing $r$ as $\mathcal{B}$'s internal randomness. In addition, for any sequence $S \in U^n$ of $n$ insertions we denote by $\hat{S}_r \subseteq U$ the set of all elements on which the $\Membership$ algorithm of $\mathcal{B}_r$ outputs $\Yes$ after processing the sequence $S$. We note that $S \in U^n$ is an ordered \emph{sequence} of (not necessarily distinct) elements, while $\hat{S}_r$ is a set. When we refer to the elements of $S$ we may abuse notation and treat $S$ as a set. Note that the fact that there are no false negatives guarantees that $S \subseteq \hat{S}_r$ for any $r \in \{0,1\}^*$. Finally, for each such $r$ and $S$ define $\mu(\hat{S}_r) = |\hat{S}_r|/u$, and define
\[ \mathcal{S}_{r, \epsilon} = \left\{S \in U^n :  \mu(\hat{S}_r)\leq 4\epsilon\right\} . \]

The following lemma uses an averaging argument and states that there exists a choice of $r \in \{0,1\}^*$ such that $\mu(\hat{S}_r)$  is rather small for many sequences $S$ (i.e., that the set $\mathcal{S}_{r, \epsilon}$ consists of many sequences).

\begin{lemma}\label{lem:S}
Let $\mathcal{B}$, $u$, $\epsilon$ and $n$ be as above. Then, there exists a string $r^* \in \{0,1\}^*$ such that $|\mathcal{S}_{r^*, \epsilon}| \ge u^n /2$.
\end{lemma}

%
%
%

\begin{proof}
The randomized data structure $\mathcal{B}$ has false positive rate at most $\epsilon$, and therefore for any sequence $S \in U^n$ it holds that
\[ \Ex_{r \leftarrow \{0,1\}^*} \left[ \mu(|\hat{S}_r|) \right] \leq \epsilon + n/u \leq 2 \epsilon . \]%
By Markov's inequality it holds that
\[ \Pr_{r \leftarrow \{0,1\}^*} \left[ \mu(| \hat S_r |)\geq 4 \epsilon \right]\leq \frac{1}{2} . \]%
In particular, there exists an $r^* \in \{0,1\}^*$ for which for at least $1/2$ of all the sequences $S \in U^n$ it holds that $\mu(| \hat S_r |) < 4 \epsilon$.
%
%
\end{proof}

\subsection{A Compression Argument for Deterministic Approximate Membership}\label{Subsec:Rand}

Form this point on focus on the deterministic data structure $\mathcal{B}_{r^*}$, where $r^* \in \{0,1\}^*$ is the internal random string $r^*$ provided by Lemma \ref{lem:S}. In this part of the proof we show that the data structure $\mathcal{B}_{r^*}$ can be used to encode the sequences in a large subset of $\mathcal{S} = \mathcal{S}_{r^*, \epsilon}$. Since Lemma~\ref{lem:S} provides a lower bound on the cardinality of $\mathcal S$, it also provides a lower bound on the length of such an encoding.

Let $S$ be a sequence in $\mathcal S$ and partition it into consecutive subsequences $S=C_1,C_2,...$ such that each $C_i$ consists of $\gamma^i$ elements, where $\gamma \geq 2$ is an integer. We define $S_i$ to be the concatenation of the first $i$ subsequences, and $n_i$ to be its length.  In other words $S_i$ is the prefix of $S$ of length $n_i = \sum_{j\leq i}\gamma^j$. Observe that since there are no false negatives, then $\hat S_i \subseteq \hat S_{i+1}$,\footnote{Recall that, as stated above, we may without loss of generality assume that $\mathcal{D}$ answers $\Yes$ on input $x$ exactly when the current state of the data structure is consistent with {\em some\/} history in which $x$ was inserted using the current array of random bits.} and therefore $\mu(\hat S_i) \leq \mu(\hat S_{i+1}) \leq 4\epsilon$ for every integer $i$.

\begin{lemma}\label{lem:i}
For any sequence $S\in \mathcal{S}$ of length $n$, there exists an integer $i$ such that  $|S_{i}|\in [\alpha n,n]$ and
\[ \mu(\hat S_{i}) - \mu(\hat S_{i-1}) \leq \frac{4\epsilon}{\log_\gamma (1/\alpha)-2} . \]
\end{lemma}
\begin{proof}
Let $j_1 = \ceil{\log_\gamma (\alpha n(\gamma - 1)+1)}$ and $j_2 = \floor{\log_\gamma( n(\gamma -1))}$. For every $j_1\leq i \leq j_2$ it holds that $n_i \in  [\alpha n,n]$. Since $\mu(\hat S_{j_1}) \geq 0$ and $\mu(\hat S_{j_2}) \leq 4\epsilon$, and since for all $i$ it holds that $\mu(\hat S_i) \leq \mu(\hat S_{i+1})$, there must be an $i \in [j_1,j_2]$ such that
\[ \mu(\hat S_{i}) - \mu(\hat S_{i-1}) \leq \frac{4\epsilon}{j_2-j_1} \leq \frac{4\epsilon}{\log_\gamma (1/\alpha)-2} . \]%
\end{proof}

Fix a sequence $S$ of length $n$, let $i$ be the smallest integer that satisfies the condition in Lemma~\ref{lem:i}, and let $k = |C_{i}\cap \hat S_{i-1}|$. That is, $k$ is the number of elements from the subsequence $C_i$ for which the $\Membership$ algorithm already answers $\Yes$ right before the $i$th subsequence $C_i$ is processed by the $\Insert$ algorithm. Observe that since the data structure is deterministic, $k=k(S)$ is completely determined by the sequence $S$.  We are interested in the case $k\leq 9\epsilon |C_i|$. In the next lemma we show that for most sequences in $\mathcal S$ this is indeed the case.

\begin{lemma}
It holds that
\begin{align}\label{lem:count}
\left| \{S\in \mathcal{S} : k(S)\leq 9\epsilon |C_i|\} \right| \geq \frac{u^n}{3} .
\end{align}
\end{lemma}

\begin{proof}
Consider a sequence $S$ which is uniformly sampled in $U^n$ one subsequence after the other. We emphasize that we sample from $U^n$  in order to avoid the dependencies associated with sampling from $\mathcal S$. Assume that each prefix  $S_j$ is associated with an arbitrary set $\hat S_j$ with measure at most $4\epsilon$. If it happens that $S_j$ is a prefix of some sequence in $\mathcal S$, then $\hat S_j$ is indeed defined as before to be the set of positive replies. Otherwise $\hat S_j$ can be any set in $U$ of measure at most $4\epsilon$.

Now, since the subsequence $C_j$ is sampled uniformly and independently from $S_{j-1}$, it holds that  $\Ex[|C_j\cap  \hat S_{j-1}|\leq 4\epsilon |C_j|]$, and by a Chernoff bound it holds that $\Pr[|C_j\cap  \hat S_{j-1}|\geq 9\epsilon |C_j|] \leq \exp(- |C_j|)$. Under our assumptions $|C_j| \geq n^{\Omega(1)}$ so by the union bound, with probability at least $1- \log_\gamma n \cdot \exp(-n^{\Omega(1)}) \geq  1-1/n$ all the $j$ for which $|C_j|$ is large enough  satisfy $|C_j\cap  \hat S_{j-1}|\geq 9\epsilon |C_j|$. Again, by the union bound we have
\begin{align*}
|\{S\in \mathcal{S}: k(S)\leq 9\epsilon |C_i|\}| \geq \left( \frac{1}{2}-\frac{1}{n} \right) u^{n} ,
\end{align*}
from which the lemma follows for all $n$ sufficiently large.
\end{proof}

Assume that after the insertion of $C_i$ the data structure uses space  $b_i$ bits. We now describe the encoding itself for a given sequence $S$.

First write the number $i$ from Lemma~\ref{lem:i}, followed by an explicit uncompressed representation of all items in the sequence $S$, except those of $C_i$. This requires at most $(n-c_i)\log u + \log\log n$ bits, where $c_i = |C_i|$. We will use the data structure in order to encode $C_i$ in a more compact form as follows. Recall that $k$ items out of $c_i$ are in $\hat S_{i-1}$. We need at most $c_i$ bits to denote where in the sequence these items are located. Next, we store the data structure itself using $b_i$ bits. We observe that since the data structure is deterministic and we write  all the elements other than $C_i$ explicitly, the encoding thus far characterizes the set $\hat S_{i-1}$. Also, since the data structure itself is written, the encoding so far characterizes the sets $\hat{S_i}$. The remaining part of the encoding consists of the elements of $C_i$ encoded relative to these two sets: We encode the $c_i-k$ elements in ${\hat S_i}\backslash \hat{S}_{i-1}$ using
$(c_i-k)\log((\mu(\hat S_i)-\mu(\hat S_{i-1}))u) + O(1)$
bits and the remaining $k$ elements using $k\log(\mu(\hat S_i)u) + O(1)$ bits.
All in all the length of this part of the encoding is:
\begin{equation}\label{eq:code}
(c_i-k)\log((\mu(\hat S_i)-\mu(\hat S_{i-1}))u) + k\log(\mu(\hat S_i)u) + O(1) .
\end{equation}%
By our choice of $i$ we have
\begin{align*}
\log(\mu(\hat S_i)-\mu(\hat S_{i-1})) \leq \log(\epsilon) - \log\log_\gamma (1/\alpha) + O(1) .
\end{align*}

Plugging in~\eqref{eq:code} and using the fact that $\mu(\hat S_i)\leq \epsilon$, the length is at most
\begin{MyEqn}
& & (c_i-k)(\log u + \log(\epsilon) - \log\log_\gamma (1/\alpha))+ k\log(\epsilon u) + O(c_i)\\
& & \qquad \qquad \leq c_i\left(\log u +\log \epsilon - (1-9\epsilon)\log\log_\gamma (1/\alpha) +O(1)\right) ,
\end{MyEqn}
and the  length of the remaining part of the encoding is at most
\begin{align*}
b_i + \log\log n + (n-c_i)\log u + c_i .
\end{align*}
By \eqref{lem:count}, the total length of the encoding  has to be greater than $\log (u^{n}/3)$ so we have:
\begin{MyEqn}
& & b_i + \log\log n + (n-c_i)\log u   + n_i\left(\log u +\log \epsilon - (1-9\epsilon)\log\log_\gamma(1/\alpha)+O(1)\right)\\
& & \qquad \qquad \geq n\log u - O(1) .
\end{MyEqn}%
which implies that
\begin{align*}
b_i \geq c_i(\log(1/\epsilon) + (1-9\epsilon)\log\log_\gamma (1/\alpha) - O(1)) .
\end{align*}
Finally, since $c_i = \gamma^i$ we have that $c_i = (n_i+\frac{1}{\gamma-1})\cdot \frac{\gamma-1}{\gamma}$ which, together with the assumption $\beta n_i \ge b_i$ in the statement of Theorem \ref{thm:lb}, completes the proof of Theorem~\ref{thm:lb}.


\section{Construction 1 (Warm-Up): Geometrically-Increasing Data Structures}\label{Sec:Construction1}

Our first construction is quite simple and natural and uses a sequence of dynamic approximate membership data structures of geometrically-increasing sizes. When viewing the sequence of inserted elements as consecutive subsequences, where the $i$th subsequence consists of $2^i$ elements, at the beginning of the $i$th subsequence we allocate and initialize a dynamic approximate membership data structure $\mathcal{B}_i$ with a false positive rate $\epsilon_i = \Theta (\epsilon / i^2)$. The elements of the $i$th subsequence are processed by the insertion algorithm of the data structure $\mathcal{B}_i$. A membership query for an element $x \in U$ is performed by invoking the membership algorithm of each of the existing data structures $\mathcal{B}_i$, and reporting $\Yes$ if any of them does. Clearly, as the underlying data structures have no false negatives, then our construction has no false negatives. In addition, a union bound guarantees that the false positive rate is at most $\sum_{i = 1}^{\infty} \epsilon_i = \Theta(\epsilon \pi^2 / 6) \leq \epsilon$ by appropriately adjusting the constants in the choices of the $\epsilon_i$.

We can instantiate the $\mathcal{B}_i$s, for example, with the dynamic approximate membership data structure resulting from the dynamic dictionary of Raman and Rao \cite{RamanR03} (via the general dictionary-based methodology described in Section \ref{SubSec:RelatedWork}). This dynamic approximate data structure supports insertions in constant expected amortized time, membership queries in constant time in the worst case, and its space consumption is $(1 + o(1)) 2^i \log(1/\epsilon_i)$ bits for any set of known size $2^i$ with a false positive rate $\epsilon_i$. This guarantees that, for any number $n$ of elements, the number of bits used by our construction after inserting any $n$ elements is at most
\begin{MyEqn}
& & (1+o(1)) n \cdot \left( \max_{1 \leq i \leq \lceil \log n \rceil} \{ \log(1/\epsilon_i) + O(1) \} \right) \\
& & \qquad = (1+o(1)) n \cdot \left( \max_{1 \leq i \leq \lceil \log n \rceil} \{ \log(1/\epsilon) + \log \left(i^2\right) + O(1) \} \right) \\
& & \qquad =  (1+o(1)) n \log(1/\epsilon) + O(n \log \log n) .
\end{MyEqn}%

Note, however, that membership queries require time $\Theta(\log n)$, since given an element $x$ we do not know to which of the $\mathcal{B}_i$ it might have been inserted in. Therefore we need a separate membership query for each of the $\mathcal{B}_i$. This yields the following theorem:
\begin{theorem}
For any $0 < \epsilon < 1$ there exists a data structure for approximate membership for sets of unknown sizes with the following properties:
\begin{enumerate}
\item The false positive rate is at most $\epsilon$.

\item For any integer $n$, the data structure uses at most $(1+o(1)) n \log(1/\epsilon) + O(n \log\log n)$ bits of space after $n$ insertions.

\item Insertions take expected amortized constant time, and for any integer $n$ membership queries are supported in $O(\log n)$ time after $n$ insertions.
\end{enumerate}
\end{theorem}

\section{Construction 2: Constant-Time Operations}\label{Sec:Construction2}

As in our first construction, when processing a sequence of elements we partition it into consecutive subsequences, where the $i$th subsequence consists of $2^i$ elements. For every integer $i$ we denote the $i$th subsequence by $s_i = x_{2^{i-1}} \cdots x_{2^i - 1}$, and denote by $S_i$ the set $\{ x_{2^{i-1}}, \ldots, x_{2^i - 1} \}$.

Let $\mathcal{H}$ be a pairwise independent collection of functions $h : U \rightarrow \{0,1\}^{\ell}$, where $\ell \ge \lceil \log(1/\epsilon) \rceil + \log u + 2$ and $|U| = u$. For each $h \in \mathcal{H}$ and integer $i \in [\ell]$ we let $h_i : U \rightarrow \{0,1\}^{\ell_i}$ be the leftmost $\ell_i = \lceil \log(1/\epsilon) \rceil + i + 2$ output bits of $h$, and let $g_i : U \rightarrow \{0,1\}^r$ be the next $r = \lceil \log \log u \rceil $ output bits of $h$ (padded with the symbol $\bot$ when less than $r$ such bits are available).


\paragraph{The basic construction.} The data structure is initialized by sampling a function $h\in \mathcal{H}$. At any point in time, when the $i$th subsequence $s_i$ is being processed, the data structure consists of a dynamic dictionary $\mathcal{D}_i$. As discussed in Section \ref{Subsec:Overview}, the insertion procedure operates in one out of two possible modes, depending on whether or not the element that is currently being inserted is the first element of its subsequence. We describe each of these modes separately.
\begin{itemize}
\item {\bf Mode 1.} When the inserted element $x \in S_i$ is not the first of its subsequence, we store the pair $(h_i(x), g_i(x))$ in the current dictionary $\mathcal{D}_i$ using $h_i(x)$ as its key.

\item {\bf Mode 2.} When the inserted element $x \in S_i$ is the first of its subsequence (i.e., $x = x_{2^{i-1}}$), we transition from the current dictionary $\mathcal{D}_{i-1}$ to a new dictionary $\mathcal{D}_i$, deallocate the space used by $\mathcal{D}_{i-1}$, and then proceed as in mode 1 above.

    Specifically, the dictionary $\mathcal{D}_i$ is initialized for storing at most $2^{i+2}$ elements, each of length $\ell_i + r$ bits. If $i > 1$ we initialize $\mathcal{D}_i$ by enumerating all pairs currently stored by $\mathcal{D}_{i-1}$, and processing each such pair $(y, \alpha_1 \cdots \alpha_r) \in \{0,1\}^{\ell_{i-1}} \times \{0,1, \bot\}^r$ as follows: If $\alpha_1 \neq \bot$, we insert to $\mathcal{D}_i$ the pair $(y \alpha_1, \alpha_2 \cdots \alpha_r \bot)$ using $y \alpha_1$ as its key. Otherwise, we insert to $\mathcal{D}_i$ the two pairs $(y0, \alpha)$ and $(y1, \alpha)$ using $y0$ and $y1$ as their keys, respectively.
\end{itemize}%

Membership queries are naturally defined: Given an element $x \in U$ and that the currently dictionary is $\mathcal{D}_i$ for some $i$, we query $\mathcal{D}_i$ with the key $h_i(x)$ to retrieve a pair of the form $(h_i(x), \alpha)$ for some $\alpha$. If such a pair is found we output $\Yes$, and otherwise we output $\No$.

\paragraph{Dealing with failures.} We note that a subtle point in the construction is that each of the dictionaries $\mathcal{D}_1, \mathcal{D}_2, \ldots$ may have a certain failure probability. Using existing dictionaries, the failure probability for each $\mathcal{D}_i$ can be made as small as any polynomial in $2^{-i}$. This means that whenever $i = \Omega(\log u)$, the failure probability can be made polynomially small in $u$, but when $i = o(\log u)$ the failure probability is rather large.

There are two standard methods for dealing with such large failure probabilities. The first is to simply rebuild each $\mathcal{D}_i$ that fails. Even for small values of $i$, the expected number of failures is typically a small constant, and thus we will be able to guarantee good expected performance. The second is to group together into one dictionary the first $u^{\delta}$ elements, for an arbitrary small constant $0 < \delta < 1$. This way, a union bound shows that no dictionary fails except with probability $u^{-c}$ for any pre-determined constant $c > 1$. For simplicity, in what follows we analyze our construction assuming that at least $n > u^{\delta}$ elements are inserted, and that we group together the first $u^{\delta}$ elements.


\paragraph{Optimal space via bucketing.} Note that the transitioning from each dictionary $\mathcal{D}_i$ to $\mathcal{D}_{i+1}$ requires storing both until all elements of $\mathcal{D}_i$ have been transitioned into $\mathcal{D}_{i+1}$ (as explained above). This increases the space used by the data structure by a multiplicative constant factor. Using a standard bucketing technique (see, for example, \cite{DietzfelbingerMadH90, DietzfelbingerR09}) we reduce the space usage of the construction when at least $n > u^{\delta}$ elements are inserted, for an arbitrary small constant $0 < \delta < 1$.

Specifically, we first hash the elements into $u^{\delta/2}$ buckets, and then apply our basic construction in each bucket. For enabling the data structure to gradually allocate more space, the data structures in the buckets are interleaved word-wise: For every $i \in [u^{\delta/2}]$, the data structure of the $i$th bucket resides in memory words whose location is equal to $i$ modulo $u^{\delta/2}$. This guarantees that if the maximum space usage of the data structures in the buckets is $s_{\rm max}$ words, then the total space required for the construction is $u^{\delta/2} \cdot s_{\rm max}$ words (and additional space can be easily allocated).

For any $u^{\delta} < n \leq u$ the hash functions of \cite{DietzfelbingerMadH90, DietzfelbingerR09} split the elements quite evenly: each bucket contains at most $(1 + o(1)) n / u^{\delta/2}$ elements, except with a probability that is polynomially small in $u$. Moreover, these functions can be evaluated in constant time. Applying our basic construction in each bucket guarantees that the transitioning operation occurs in at most one bucket at any point in time, and therefore the additional space that is required is proportional to the number of elements in each bucket and not to total number of elements.

\paragraph{Performance analysis.} The following theorem is obtained by instantiating our construction with a sufficiently good construction of a dynamic dictionary. For example, the dynamic dictionary of Raman and Rao \cite{RamanR03} is space optimal (up to additive lower-order terms), supports insertions in constant expected amortized time, and membership queries in constant time in the worst case.
\begin{theorem}\label{Theorem:Construction2}
For any $0 < \epsilon < 1$, integer $u$, and constant $c > 1$, there exists a data structure for approximate membership for sets of unknown sizes from a universe of size $u$ with the following properties:
\begin{enumerate}
\item The false positive rate is $\epsilon + u^{-c}$.

\item For any constant $0 < \delta < 1$ and $n > u^{\delta}$, the data structure uses at most $(1 + o(1)) n \log(1/\epsilon) + O(n \log\log n)$ bits of space after $n$ insertions.

\item Insertions take expected amortized constant time, and membership queries take constant time in the worst case.
\end{enumerate}
\end{theorem}

\begin{proof}
As discussed above, hashing the inserted elements into $u^{\delta/2}$ buckets results in a balanced allocation up to additive lower order terms with all but a polynomially small probability in $u$. Therefore, for simplicity, from this point on we focus on $n$ elements that are inserted into a single bucket. We first prove that for every $i$, at most $2^{i+2}$ elements are inserted into the dictionary $\mathcal{D}_i$. Fix an $i$, and partition the elements that are inserted to $\mathcal{D}_i$ to two disjoint sets: elements that correspond to elements from $S_1, \ldots, S_{i-r}$, and elements that correspond to elements from $S_{i-r+1}, \ldots, S_i$. For each element $x$ that belongs to some $S_j$, we observe that it contributes $2^{i - j - r}$ elements if $1 \leq j \leq i - r$, and exactly one element if $i - r + 1 \leq j \leq i$. Therefore, the number of elements that are inserted into $\mathcal{D}_i$ is
\begin{MyEqn}
\sum_{j = 1}^{i-r} |S_j| \cdot 2^{i-j-r} + \sum_{j = i - r + 1}^i |S_j| & = & \sum_{j = 1}^{i-r} 2^{i-r-1} + \sum_{j = i - r + 1}^i 2^j \\
& \leq & 2^{i+2} .
\end{MyEqn}%

Now, for bounding the false positive rate, fix a sequence $x_1 \cdots x_j$ of inserted elements, an element $x \notin \{x_1, \ldots, x_j\}$, and let $i$ be such that $2^{i-1} \leq j \leq 2^i- 1$. Then, the current state of the data structure consists of a dictionary $\mathcal{D}_i$, and a query for $x$ initiates a membership query for the key $h_i(x)$. Since at most $2^{i+2}$ keys were inserted so far to the dictionary $\mathcal{D}_i$, the pairwise independence of $\mathcal{H}$ guarantees that $x$ forms a collision with some existing element with probability at most $2^{i+2} \cdot 2^{- \ell_i}  \leq \epsilon$. In addition, we assume that the constructions of all the $\mathcal{D}_i$ are successful except with probability $u^{-c}$, and therefore the false positive rate is at most $\epsilon + u^{-c}$.

We now bound the space overhead. Assume that $2^{i-1} \leq n \leq 2^i - 1$ elements were inserted, and that the current dictionary $\mathcal{D}_i$ is constructed using a dictionary that can store $n$ elements from a universe of size $u' = {\rm poly}(n')$ with $r$ bits of satellite data using space $ (1 + o(1)) n (\log(u'/n) + r)$ bits (e.g., \cite{RamanR03} as discussed above). Then, the space utilized by $\mathcal{D}_i$ is at most
\begin{MyEqn}
(1 + o(1)) n ( \log( 2^{\ell_i}/2^{i}) + r) & \leq & (1 + o(1)) n \log(1/\epsilon) + O(n \log \log u) \\
& = & (1 + o(1)) n \log(1/\epsilon) + O(n \log \log n)
\end{MyEqn}%
Finally, note that membership queries are supported in constant time, and that the expected amortized insertion time is also constant (as in the underlying dictionary).
\end{proof}

In the remainder of this section we describe two extensions of our construction. The first extension shows how to enjoy constant-time insertions in the worst case by increasing the space usage from $(1 + o(1)) n \log(1/\epsilon) + O(n \log\log n)$ to $O(n \log(1/\epsilon) + n \log\log n)$. The second extension shows how to support deletions (which have to be carefully defined).

\paragraph{Constant-time insertions in the worst case via de-amortization.} As presented above using the two different insertion modes, the construction enjoys a good amortized insertion time: Most insertions correspond to mode 1 and are processed very fast, while only a small number of insertions correspond to case 2. The main observation is that if the underlying $\mathcal{D}_i$ offers a constant insertion time in the worst case with high probability (e.g., as in \cite{ArbitmanNS10}), then our construction {\em without the bucketing} can be de-amortized: Instead of initializing each $\mathcal{D}_i$ only when inserting $x_{2^{i-1}}$, then the total amount of work required for initializing $\mathcal{D}_i$ can be equally split among the insertions of $x_{2^{i-1}}, \ldots, x_{2^i - 1}$. Specifically, on each such insertion, devote a constant number of additional steps for the initialization of $\mathcal{D}_i$. As shown in the proof of Theorem \ref{Theorem:Construction2}, for every $i$ at most $2^{i+1}$ elements are inserted into the dictionary $\mathcal{D}_{i-1}$. Therefore, the total amount of work (in the worst case) required for initializing $\mathcal{D}_i$ is $O(2^i)$. We note that the idea of bucketing the elements that we used above does not seem useful here. The reason is that it is no longer the case that a transition between dictionaries occurs in at most one bucket at any point in time. Therefore, the space usage (even with bucketing) would be $O(n \log(1/\epsilon) + n \log\log n)$ bits (with a rather small leading constant) instead of $(1 + o(1)) n \log(1/\epsilon) + O(n \log\log n)$ bits as in Theorem \ref{Theorem:Construction2}.

\paragraph{Supporting deletions.} Note that for any approximate membership data structure it is impossible to detect if an attempt is made to delete a false positive. Thus, the data structure must put the burden on its user to ensure that deletions are applied only to elements that are in fact in the set (if this contract is broken, false negatives may arise). In the space analysis we will also assume that insertions are proper, i.e., an element may be inserted at most once.

A well-known approach to supporting deletions~\cite{PaghPR05} is to store the {\em multiset\/} of signatures rather than just the set of distinct signatures. A deletion of an element with signature $h(x)$ is implemented by decreasing the multiplicity of $h(x)$ in the multiset by 1. However, there are  complications when trying to make this technique work in our setting. For example, in Construction 1, the element to be deleted may be a false positive in one of the data structures $\mathcal{B}_i$ and a ``true positive'' in another data structure $\mathcal{B}_j$. The problem is that there is no way to tell which is the false positive, and if we remove $h(x)$ from $\mathcal{B}_i$ a false negative will occur. A similar problem occurs in Construction 2, where it may not be possible to determine which signatures are to be deleted.

Our way around this problem is to abandon the idea of storing a multiset of signatures, but rather use a secondary dictionary data structure whenever we encounter identical signatures. In the following we describe how to augment Construction 2 with deletions. When inserting an element $x$ we first check if it is a false positive of the existing set. Every false positive is inserted in the secondary data structure, while remaining elements are inserted in the primary data structure. Membership queries are extended to also look up the element in the secondary data structure, which has zero false positive rate. The deletion algorithm first checks if the element can be deleted from the secondary data structure. If not, its signature(s) need to be deleted from the primary structure. However, elements that were inserted when the set was smaller may be associated with a large number of signatures, formed by extending an original signature with all possible bit strings to form a set of possible signatures matching the current signature length. To allow efficient deletion, we extend the information that is stored with each key with the length of the original signature, and with a bit that can be used to indicate deletion. Deletion is performed by marking the lexicographically smallest signature in the set (i.e., the one extended with only zeros) as deleted. The membership procedure is then modified to compute this signature, and check whether it has been marked as deleted. To ensure that we do not use significant space for signatures of deleted keys, we run a background process that periodically checks if each signature can be removed from the set, spending constant time per update. In a similar way, we periodically check keys in the secondary structure to see if they remain false positives, or can be moved to the primary structure.

It is easy to see that the data structure will work correctly (under the assumption of proper deletions and insertions). What is less obvious is how much extra space is needed for the secondary structure. Observe that we may without loss of generality assume that the false positive rate is at most $1/\log u$, since we allow a space overhead of $O(n\log\log n)$ bits, and $n>u^\delta$. This means that the expected number of false positives in a set of $n$ elements is $O(n/\log u)$, so storing this set requires just $O(n)$ bits in expectation. To ensure a high probability bound on the space usage, we need a stronger hash function to compute the signatures. In particular, from~\cite{DietzfelbingerGMP92} it follows that using constant-degree polynomial hash functions we can ensure that the number of signature collisions, corresponding to false positives, will be within a constant factor of the expectation with probability $1-u^{-c}$, for any desired constant $c$.

\section{Directions for Future Research}\label{Sec:FutureResearch}

Our work raises several fundamental directions for future research both from a theoretical perspective and from a practical perspective. From a theoretical perspective, an interesting problem is to tighten our lower bound by identifying the leading constant in the additive $\Omega(n \log\log n)$ factor. In addition, it would be interesting to explore whether our constructions can be improved by a data structure that simultaneously enjoys the best of both worlds: space consumption of $(1 + o(1)) n \log(1/\epsilon) + O(n \log\log n)$ bits and constant-time operations in the worst case with high probability.

From a more practical perspective, while Bloom filters \cite{Bloom70} provide a practical solution in the setting where an upper bound $n$ is known in advance, our cosntruction do not seem to enjoy the same level of practicality in the setting where such an upper is not known in advance. Specifically, our first construction supports membership queries in time $O(\log n)$, which may be too slow in some applications, and our second construction suffers from non-trivial hidden constants due to our de-amortization technique. It would be very interesting to design a practical solution that matches our space lower bound.


\bibliographystyle{newbib2}
\bibliography{BIB}

\newcommand{\etalchar}[1]{$^{#1}$}
\begin{thebibliography}{DMadH90}

\bibitem[ABP{\etalchar{+}}07]{AlmeidaBPH07}
P.~S. Almeida, C.~Baquero, N.~Pregui{\c{c}}a, and D.~Hutchison.
\newblock Scalable {B}loom filters.
\newblock {\em Information Processing Letters}, 101(6):255--261, 2007.

\bibitem[ANS09]{ArbitmanNS09}
Y.~Arbitman, M.~Naor, and G.~Segev.
\newblock De-amortized cuckoo hashing: {P}rovable worst-case performance and
  experimental results.
\newblock In {\em Proceedings of the 36th International Colloquium on Automata,
  Languages and Programming}, pages 107--118, 2009.

\bibitem[ANS10]{ArbitmanNS10}
Y.~Arbitman, M.~Naor, and G.~Segev.
\newblock Backyard cuckoo hashing: Constant worst-case operations with a
  succinct representation.
\newblock In {\em Proceedings of the 51th Annual IEEE Symposium on Foundations
  of Computer Science}, pages 787--796, 2010.

\bibitem[Blo70]{Bloom70}
B.~H. Bloom.
\newblock Space/time trade-offs in hash coding with allowable errors.
\newblock {\em Communications of the ACM}, 13(7):422--426, 1970.

\bibitem[BM03]{BroderM03}
A.~Z. Broder and M.~Mitzenmacher.
\newblock Network applications of {B}loom filters: A survey.
\newblock {\em Internet Mathematics}, 1(4), 2003.

\bibitem[BMP{\etalchar{+}}06]{BonomiMPSV06}
F.~Bonomi, M.~Mitzenmacher, R.~Panigrahy, S.~Singh, and G.~Varghese.
\newblock An improved construction for counting {Bloom} filters.
\newblock In {\em Proceedings of the 14th Annual European Symposium on
  Algorithms}, pages 684--695, 2006.

\bibitem[BV13]{BelazzouguiV13}
D.~Belazzougui and R.~Venturini.
\newblock Compressed static functions with applications to other dictionary
  problems.
\newblock In {\em Proceedings of the 24th Annual ACM-SIAM Symposium on Discrete
  Algorithms}, 2013.

\bibitem[CFG{\etalchar{+}}78]{CarterFGMW78}
L.~Carter, R.~W. Floyd, J.~Gill, G.~Markowsky, and M.~N. Wegman.
\newblock Exact and approximate membership testers.
\newblock In {\em Proceedings of the 10th Annual ACM Symposium on Theory of
  Computing}, pages 59--65, 1978.

\bibitem[CKR{\etalchar{+}}04]{ChazelleKRT04}
B.~Chazelle, J.~Kilian, R.~Rubinfeld, and A.~Tal.
\newblock The {Bloomier} filter: An efficient data structure for static support
  lookup tables.
\newblock In {\em Proceedings of the 15th Annual ACM-SIAM Symposium on Discrete
  Algorithms}, pages 30--39, 2004.

\bibitem[CM03]{CohenM03}
S.~Cohen and Y.~Matias.
\newblock Spectral {B}loom filters.
\newblock In {\em Proceedings of the 2003 {ACM} {SIGMOD} International
  Conference on Management of Data}, pages 241--252, 2003.

\bibitem[CW79]{CarterW79}
L.~Carter and M.~N. Wegman.
\newblock Universal classes of hash functions.
\newblock {\em Journal of Computer and System Sciences}, 18(2):143--154, 1979.

\bibitem[DGM{\etalchar{+}}92]{DietzfelbingerGMP92}
M.~Dietzfelbinger, J.~Gil, Y.~Matias, and N.~Pippenger.
\newblock Polynomial hash functions are reliable.
\newblock In {\em Proceedings of the 19th International Colloquium on Automata,
  Languages and Programming}, pages 235--246, 1992.

\bibitem[DMadH90]{DietzfelbingerMadH90}
M.~Dietzfelbinger and F.~Meyer auf~der Heide.
\newblock A new universal class of hash functions and dynamic hashing in real
  time.
\newblock In {\em Proceedings of the 17th International Colloquium on Automata,
  Languages and Programming}, pages 6--19, 1990.

\bibitem[DP08]{DietzfelbingerP08}
M.~Dietzfelbinger and R.~Pagh.
\newblock Succinct data structures for retrieval and approximate membership.
\newblock In {\em Proceedings of the 35th International Colloquium on Automata,
  Languages and Programming}, pages 385--396, 2008.

\bibitem[DR09]{DietzfelbingerR09}
M.~Dietzfelbinger and M.~Rink.
\newblock Applications of a splitting trick.
\newblock In {\em Proceedings of the 36th International Colloquium on Automata,
  Languages and Programming}, pages 354--365, 2009.

\bibitem[FCA{\etalchar{+}}00]{FanCAB00}
L.~Fan, P.~Cao, J.~M. Almeida, and A.~Z. Broder.
\newblock Summary cache: {A} scalable wide-area web cache sharing protocol.
\newblock {\em IEEE/ACM Transactions on Networking}, 8(3):281--293, 2000.

\bibitem[GWC{\etalchar{+}}06]{GuoWCL06}
D.~Guo, J.~Wu, H.~Chen, and X.~Luo.
\newblock Theory and network applications of dynamic bloom filters.
\newblock In {\em INFOCOM 2006. 25th IEEE International Conference on Computer
  Communications. Proceedings}, pages 1--12, 2006.

\bibitem[GWC{\etalchar{+}}10]{GuoWCYL10}
D.~Guo, J.~Wu, H.~Chen, Y.~Yuan, and X.~Luo.
\newblock The dynamic {B}loom filters.
\newblock {\em IEEE Transactions on Knowledge and Data Engineering},
  22(1):120--133, 2010.

\bibitem[Hag98]{Hagerup98}
T.~Hagerup.
\newblock Sorting and searching on the word {RAM}.
\newblock In {\em Proceedings of the 15th Annual Symposium on Theoretical
  Aspects of Computer Science}, pages 366--398, 1998.

\bibitem[HKL08]{HaoKL08}
F.~Hao, M.~S. Kodialam, and T.~V. Lakshman.
\newblock Incremental {B}loom filters.
\newblock In {\em Proceedings of the 27th IEEE International Conference on
  Computer Communications}, pages 1067--1075, 2008.

\bibitem[HMP01]{HagerupMP01}
T.~Hagerup, P.~B. Miltersen, and R.~Pagh.
\newblock Deterministic dictionaries.
\newblock {\em Journal of Algorithms}, 41(1):69--85, 2001.

\bibitem[LP10]{LovettP10}
S.~Lovett and E.~Porat.
\newblock A lower bound for dynamic approximate membership data structures.
\newblock In {\em Proceedings of the 51th Annual IEEE Symposium on Foundations
  of Computer Science}, pages 797--804, 2010.

\bibitem[Mil99]{Miltersen99}
P.~B. Miltersen.
\newblock Cell probe complexity -- {A} survey.
\newblock In {\em Proceedings of the 19th Conference on the Foundations of
  Software Technology and Theoretical Computer Science, Advances in Data
  Structures Workshop}, 1999.

\bibitem[MPP05]{MortensenPP05}
C.~W. Mortensen, R.~Pagh, and M.~P{\u{a}}tra{\c{s}}cu.
\newblock On dynamic range reporting in one dimension.
\newblock In {\em Proceedings of 37th Annual ACM Symposium on Theory of
  Computing}, pages 104--111, 2005.

\bibitem[Por09]{Porat09}
E.~Porat.
\newblock An optimal {Bloom} filter replacement based on matrix solving.
\newblock In {\em Proceedings of the 4th International Computer Science
  Symposium in {Russia}}, pages 263--273, 2009.

\bibitem[PP08]{PaghP08}
A.~Pagh and R.~Pagh.
\newblock Uniform hashing in constant time and optimal space.
\newblock {\em SIAM Journal on Computing}, 38(1):85--96, 2008.

\bibitem[PPR05]{PaghPR05}
A.~Pagh, R.~Pagh, and S.~S. Rao.
\newblock An optimal {B}loom filter replacement.
\newblock In {\em Proceedings of the 16th Annual ACM-SIAM Symposium on Discrete
  Algorithms}, pages 823--829, 2005.

\bibitem[RR03]{RamanR03}
R.~Raman and S.~S. Rao.
\newblock Succinct dynamic dictionaries and trees.
\newblock In {\em Proceedings of the 30th International Colloquium on Automata,
  Languages and Programming}, pages 357--368, 2003.

\bibitem[WJZ{\etalchar{+}}11]{WeiJZF11}
J.~Wei, H.~Jiang, K.~Zhou, and D.~Feng.
\newblock {DBA}: {A} dynamic {B}loom filter array for scalable membership
  representation of variable large data sets.
\newblock In {\em Proceedings of the 19th Annual IEEE/ACM International
  Symposium on Modeling, Analysis and Simulation of Computer and
  Telecommunication Systems}, pages 466--468, 2011.

\bibitem[WJZ{\etalchar{+}}13]{WeiJZF13}
J.~Wei, H.~Jiang, K.~Zhou, and D.~Feng.
\newblock Efficiently representing membership for variable large data sets.
\newblock To appear in {\em IEEE Transactions on Parallel and Distributed
  Systems}, 2013.

\end{thebibliography}

\end{document}